\documentclass[11pt]{article}
\usepackage{amsmath,amssymb,amsthm}
\usepackage{fullpage}
\usepackage[colorlinks]{hyperref}
\hypersetup{linkcolor=cyan,filecolor=cyan,citecolor=cyan,urlcolor=cyan}
\usepackage{xspace}
\usepackage{thm-restate,color,xcolor}
\usepackage{boxedminipage}
\usepackage[boxed]{algorithm}
\usepackage{epigraph}
\usepackage{amsmath}
\usepackage{mathtools}
\usepackage{framed}
\usepackage[framemethod=tikz]{mdframed}
\usepackage{titlesec}
\usepackage{lipsum}%
\usepackage{cleveref,aliascnt}
\usepackage{tikz}
\usepackage{wrapfig}
\usepackage{framed}
\usepackage[framemethod=tikz]{mdframed} 

\usetikzlibrary{shapes.geometric}
\usetikzlibrary{arrows}
\usetikzlibrary{arrows.meta}
\usetikzlibrary{patterns}
\usetikzlibrary{shapes.misc}

\theoremstyle{plain}
\newtheorem{thm}{Theorem}[section]
\newcommand{\BTHM}{\begin{thm}} \newcommand{\ETHM}{\end{thm}}
\newtheorem{cor}[thm]{Corollary}
\newcommand{\BCR}{\begin{cor}} \newcommand{\ECR}{\end{cor}}
\newtheorem{lem}[thm]{Lemma}
\newcommand{\BL}{\begin{lem}}   \newcommand{\EL}{\end{lem}}
\newtheorem{clm}[thm]{Claim}
\newcommand{\BCM}{\begin{clm}}   \newcommand{\ECM}{\end{clm}}
\newtheorem{prop}[thm]{Proposition}
\newcommand{\BP}{\begin{prop}}   \newcommand{\EP}{\end{prop}}
\newtheorem{assm}[thm]{Assumption}
\newcommand{\BASM}{\begin{assm}}   \newcommand{\EASM}{\end{assm}}

\theoremstyle{definition}
\newtheorem{defn}{Definition}[section]
\newcommand{\BD}{\begin{defn}}   \newcommand{\ED}{\end{defn}}
\newtheorem{con}[thm]{Conjecture}
\newcommand{\BCONJ}{\begin{con}}   \newcommand{\ECONJ}{\end{con}}

\theoremstyle{definition}
\newtheorem{problem}[thm]{Problem}
\newcommand{\BPR}{\begin{problem}}   \newcommand{\EPR}{\end{problem}}

\newenvironment{rem}{\noindent{\bf Remark:~~}}{}
\newcommand{\BREM}{\begin{rem}} \newcommand{\EREM}{\end{rem}}
\newenvironment{discussion}{\noindent{\bf Discussion:~~\\}}{}
\newcommand{\BDIS}{\begin{discussion}} \newcommand{\EDIS}{\end{discussion}}

\numberwithin{equation}{section}

\usepackage{wrapfig}

\usepackage{changepage}

\def\blackslug
     {\hbox{\hskip 1pt\vrule width 8pt height 8pt depth 1.5pt\hskip 1pt}}
\def\qed{\quad\blackslug\lower 8.5pt\null\par}

\newcommand{\Direachability}{\mathsf{Direachability}}

\newcommand{\OutSucc}[2]      {\mathcal{S}\ifblank{#1}{(#2)}{(#2,#1)}}

\usepackage{enumitem}

\usepackage{subcaption} 
\usepackage{booktabs}

\newtheorem{exmp}[thm]{Example}
\newtheorem{fact}[thm]{Fact}
\newcommand{\BEX}{\begin{exmp}} \newcommand{\EEX}{\end{exmp}}
\newcommand{\BF}{\begin{fact}}   \newcommand{\EF}{\end{fact}}
\newcommand{\Bcr}{\begin{techcorr}}
\newcommand{\Ecr}{\end{techcorr}}
\newcommand{\BDS}{\begin{description}}
\newcommand{\EDS}{\end{description}}
\newcommand{\BE}{\begin{enumerate}}
\newcommand{\EE}{\end{enumerate}}
\newcommand{\BI}{\begin{itemize}}
\newcommand{\EI}{\end{itemize}}
\newcommand{\BPF}{\begin{proof}}
\newcommand{\EPF}{\end{proof}}

\newcommand{\BB}{\begin{enumerate}}
\newcommand{\EB}{\end{enumerate}}

\usepackage{booktabs}
\usepackage{tabularx}

\newenvironment{wrapper}[1]
{
	\begin{center}
		\begin{minipage}{\linewidth}
			\begin{mdframed}[hidealllines=true, backgroundcolor=gray!20, leftmargin=0cm,innerleftmargin=0.5cm,innerrightmargin=0.5cm,innertopmargin=0.5cm,innerbottommargin=0.5cm,roundcorner=10pt]
				#1}
			{\end{mdframed}
		\end{minipage}
	\end{center}
} 

\usepackage{geometry}
\title{Multi-Source Reachability in Near-Optimal Time}
\author{
Shimon Kogan 
        \small Weizmann Institute\\
        \small shimon.kogan@weizmann.ac.il
\and				
Merav Parter \thanks{This project is funded by the European Research Council (ERC) under the European Union Horizon 2020 research and innovation programme (grant agreement No. 949083).}
        \small Weizmann Institute \\
        \small merav.parter@weizmann.ac.il
}
\date{}

\begin{document}
\maketitle

\begin{abstract}
The multi-source reachability problem asks to compute the reachable sets from a given subset of source vertices. For $n$-vertex digraphs $G=(V,E)$ and a subset of sources $S \subseteq V$ with $|S|=n^{\sigma}$ for some $\sigma \in [0,1]$, we present a near-optimal deterministic algorithm that solves this problem in $\tilde{O}(n^{\omega(\sigma)})$ time, where $\omega(\sigma)$ is the rectangular matrix multiplication exponent for multiplying
an $n^{\sigma}\times n$ matrix by an $n \times n$ matrix. For dense graphs, this yields reachability from up to $n^{0.32}$ sources in near-linear time, breaking the super-quadratic time barrier and improving over the state-of-the-art $n^{1+2/3\omega(\sigma)}$-time randomized algorithm of Elkin and Trehan [arXiv:2401.05628, 2024].
%
%

\end{abstract}


\pagenumbering{gobble}

\pagenumbering{arabic}

\section{Introduction}

We present a near-optimal-time algorithm for computing multi-source reachability in directed graphs. In the multi-source reachability problem, given a digraph $G=(V,E)$ and a subset of sources $S \subseteq V$, one must compute, for each source $s$, all vertices that are reachable from $s$. 
Despite extensive research over the years \cite{UllmanY91,Cohen1996SeparatorSP}, no known near-optimal algorithm for this problem improves over either running BFS/DFS from each source or computing the full transitive closure. Kogan and Parter \cite{FasterKoganP23} showed that there is not much hope for improving the combinatorial time bound to $o(|S|\cdot n^{2-o(1)})$, assuming the combinatorial BMM conjecture. Elkin and Trehan \cite{ElkinTrehan24} recently presented an $n^{1+2/3\omega(\sigma)}$-time algebraic algorithm using recent constructions of reachability shortcuts \cite{KoganPSODA22}. We show:

\BTHM\label{thm:multi-source}
    There is a deterministic algorithm that, given an $n$-vertex digraph $G=(V,E)$ and a subset of sources $S \subseteq V$ with $|S|=n^{\sigma}$, computes multi-source reachability in $\tilde{O}(n^{\omega(\sigma)})$ time, where $\omega(\sigma)$ is the rectangular matrix multiplication exponent for multiplying an $n^{\sigma}\times n$ matrix by an $n \times n$ matrix. This time complexity is near-optimal.
\ETHM

To show tightness of our time complexity, we use a simple reduction from multiplying an $n^{\sigma}\times n$ Boolean matrix by an $n \times n$ Boolean matrix. Consider a three-layered graph $(S,V_1,V_2)$ with $|S|=n^{\sigma}$, where the edges of $S \times V_1$ are determined by $A$ and the edges of $V_1 \times V_2$ are determined by $B$. Then computing multi-source reachability, represented by an $n^{\sigma} \times n$ matrix $C$, corresponds to computing $A \times B$. 
 
For various algorithmic applications, it is useful to also obtain the predecessor $S \times V$ matrix, which allows one to recover the dipaths in time linear in their lengths:
 
\smallskip
\noindent \textbf{Witness and Predecessor Matrices.} Given an $a \times b$ Boolean matrix $A$ and a $b \times c$ Boolean matrix $B$, a witness matrix for the matrix multiplication $C=A \times B$ is an $a \times c$ matrix $W$ defined as follows: If $C(i,j)=0$, set $W(i,j)=0$; if $C(i,j)=1$, then $W(i,j)$ is some index $k$ such that $A(i,k)=1$ and $B(k,j)=1$.
For an $a \times n$ matrix $R$ that represents the reachability from $a$ vertices $S=\{s_1,\ldots, s_a\}$ to all vertices $V=\{v_1,\ldots, v_n\}$ in a base graph $G=(V,E)$, the $a \times n$ predecessor matrix $P$ specifies the predecessor neighbor on some $s_i$ to $v_j$ path, if one exists, for every $s_i,v_j \in S \times V$. Formally, for every $s_i,v_j \in S \times V$ with $R(i,j)=1$, $P(i,j)=k$ where $(v_k,v_j)\in E(G)$ and $R(i,k)=1$. If $R(i,j)=0$, then $P(i,j)=0$.

Alon et al. \cite{AlonGMN92} presented a deterministic algorithm for computing the witness matrix in nearly the same time as the matrix multiplication.
\begin{thm}\label{thm:witness-det}
    There is a deterministic algorithm that, given binary matrices $A$ and $B$ of dimensions $n^{\sigma} \times n$ and $n\times n$, respectively, computes the $n^{\sigma} \times n$ witness matrix for the $A \times B$ multiplication in $\tilde{O}(n^{\omega(\sigma)})$ time.
\end{thm}

\section{Multi-Source Reachability in Near-Optimal Time}

\textbf{Notation.} For a vertex $v$, let $\OutSucc{G}{v}$ be the successors of $v$ in $G$ (i.e., the vertices that are reachable from $v$ in $G$). Throughout, reachability is non-reflexive in the output: the trivial pair $(s,s)$ is not included merely because of the empty path. 
For a given source set $S \subseteq V$, our objective is to output the sets $\{\OutSucc{G}{s}\}_{s \in S}$. For ease of notation, for vertex subsets $V_1,V_2$ and a graph $G'$ with $V_1,V_2 \subseteq V(G')$, let $E_{G'}[V_1 \to V_2]$ be an $|V_1| \times |V_2|$ matrix that represents the $V_1 \times V_2$ edges in $G'$. Formally, letting $V_1=[a_1,\ldots,a_{\ell_1}]$ and $V_2=[b_1,\ldots,b_{\ell_2}]$ be the vertex subsets ordered by their topological ordering in $G'$, then 
\begin{equation}\label{eq:edge-matrix}
E_{G'}[V_1 \to V_2](i,j)=\begin{cases}
			1, & \text{if $(a_i,b_j)\in E(G')$}\\
      0, & \text{otherwise.}
		 \end{cases}
\end{equation}
Given two matrices $A, B \in \mathbb{R}^{a \times b}$, let $C=A \circ B$ be the $a \times 2b$ matrix obtained by concatenating the columns of matrix $B$ to those of matrix $A$. That is, the first (resp., last) $b$ columns of $C$ correspond to matrix $A$ (resp., $B$). For a given $n^{\sigma}\times n$ matrix $P$ where the rows correspond to $k=n^{\sigma}$ sources $S=\{s_1,\ldots, s_k\}$ and the columns to $n$ vertices $V=\{v_1,\ldots, v_n\}$, we sometimes override notation and write $P(s,v)$ for the matrix entry in the designated row and column of $s$ and $v$, respectively; that is, for $s=s_i$ and $v=v_j$, we have $P(s,v)=P(i,j)$.

\paragraph{Informal Description.} We assume, for clarity of exposition, that the input graph $G$ is a DAG; at the end of the section, we describe how to handle general graphs. 
The algorithm is recursive and works as follows. Let $V=[v_1,\ldots, v_n]$ be the topologically ordered set of vertices. Split $V$ into two equal-size parts, $V_\ell=[v_1,\ldots, v_{n/2}]$ and $V_r=[u_1=v_{n/2+1},\ldots, u_{n/2}=v_{n}]$, and assume that the algorithm has already computed the $S \times V_{\ell}$ reachability matrix, $R_{\ell}$, in the graph $G_{\ell}=G[V_{\ell}]$. By the topological 
ordering, any $S \times V_\ell$ dipath\footnote{I.e., a $u$-$v$ dipath for $u \in S$ and $v \in V_{\ell}$.} in $G$ must be contained in $G_\ell$. 
The goal is then to apply the algorithm recursively to compute the $S \times V_r$ reachability in some $(n/2+|S|)$-vertex graph $G_{r}$. To define $G_r$, one needs to account for the $S \times V_r$ paths that go through the $V_{\ell}$ vertices. This is done by solving the following matrix multiplication: $R'=R_{\ell}\times E_G[V_\ell \to V_r]$. Since $R_{\ell}$ is an $|S| \times n/2$ matrix and $E_G[V_\ell \to V_r]$ is an $n/2 \times n/2$ matrix, the computation of $R'$ takes $O(n^{\omega(\sigma)})$ time. The matrix $R'$ satisfies the following property: $R'(i,j)=1$ iff there is an $s_i$ to $u_j$ dipath in $G$ all of whose internal vertices are in $V_\ell$. The desired graph $G_r$ is then obtained by adding the edges defined by the matrix $R'$ to the graph $G[V_r]$. Namely, $E(G_r)=E(G[V_r]) \cup \{(s_i,u_j)~\mid~ R'(i,j)=1\}$. The algorithm then recurs on the instance $(S, V_r, G_r)$.

In particular, $G_r$ is not merely the induced graph $G[V_r]$; it is $G[V_r]$ augmented with shortcut edges from the sources in $S$ to vertices in $V_r$. These shortcut edges summarize the paths whose internal vertices lie in $V_\ell$.

To analyze the runtime, let $T(n)$ be the time complexity for computing multi-source reachability from $S$ sources to $n$ vertices. We then have $T(n)\leq 2T(n/2)+MM(|S|,n/2)$, where $MM(a,b)$ is the complexity of computing the matrix multiplication of an $a \times b$ matrix with a $b \times b$ matrix. Using the inequality $n^{\sigma} \leq 2^j \cdot (n/2^j)^{\sigma}$, solving this recurrence gives $T(n)=\tilde{O}(n^{\omega(\sigma)})$. We next provide the detailed algorithm description.

%
%




\paragraph{Algorithm $\Direachability$ (for DAGs).} The input is an $n$-vertex DAG $G$, a source set $S \subseteq V(G)$, and a subset of vertices $V'$ (initially, $V'=V$). The algorithm is also given the topological ordering of $V$ in $G$, and it treats both $V$ and $S$ as ordered sets according to this topological ordering. 
We assume, w.l.o.g., that $|V(G)|=2^k$ for some integer\footnote{Clearly this can be assumed by adding isolated vertices which increases the number of vertices by a factor of at most $2$.} $k$. 

The algorithm is recursive and works as follows. If $|V'|=1$, where $V'=\{v\}$, the algorithm outputs an $S \times V'$ reachability matrix according to the sources in $S$ that have an edge into $v$. Otherwise, the algorithm partitions $V'$ into two equal sets $V_\ell$ and $V_r$. It recursively computes the $S \times V_\ell$ reachability matrix in the graph $G_\ell=G[V_\ell]$, denoted by $R_\ell$. 
Observe that, by the ordering, for any $s\in S$ and $u \in V_\ell$, every $s$ to $u$ path in $G$ must be fully contained in $G[V_\ell]$. This no longer holds for $V_r$, since $s$ to $V_r$ paths may visit vertices in $V_\ell$. To handle this, the algorithm defines a graph $G_r=(V_r \cup S, E(G[V_r]) \cup E')$ over the vertices in $V_r \cup S$. The edge set $E'$ contains all edges $(s_i,u) \in S \times V_r$ such that there is an $s_i$ to $u$ path in $G$ whose internal vertices are in $V_\ell$. This edge set $E'$ is obtained by computing the matrix multiplication $R_{\ell} \times E_G[V_{\ell} \to V_r]$. Thus, $G_r$ is not the induced graph $G[V_r]$; it is $G[V_r]$ augmented with shortcut edges from the sources to vertices in $V_r$. Next, the algorithm recursively computes the $S \times V_r$ reachability matrix in $G_r$, denoted by $R_r$. Finally, the output $S\times V$ reachability matrix $R$ is obtained by concatenating the matrices $R_\ell$ and $R_r$. See below for formal pseudocode.

\begin{wrapper}\vspace{-2pt}
\begin{center}\textbf{Algorithm $\Direachability(G,S,V')$}
\end{center}
\begin{enumerate}
\item Let $V'=[v'_1,\ldots,v'_{|V'|}]$ and $S=[s_1,\ldots, s_\ell]$ be the topological ordering of $V'$ and $S$.
\begin{itemize}
\item If $|V'|=1$ do:
\begin{itemize}
\item Set $R(i,1)=1$ for every $s_i \in N^{in}_G(v'_1)$ and $P(i,1)=i$.
\item Return $R$ and $P$.
\end{itemize}
\end{itemize}

\item Let $V_{\ell}=[v'_1,\ldots, v'_{|V'|/2}]$ and $V_r=[u'_1=v'_{|V'|/2+1},\ldots, u'_{|V'|/2}=v'_{|V'|}]$.
\item Set $(R_{\ell},P_{\ell})=\Direachability(G[V_{\ell}],S,V_{\ell})$.
\item Let $Z_{\ell \to r}=R_{\ell} \times E_G[V_{\ell} \to V_r]$ and $W_{\ell \to r}$ be the corresponding witness matrix.
\item Let $E_{\ell \to r}=\{(s_i,u'_j) ~\mid~ Z_{\ell \to r}(i,j)=1\} \setminus E(G)$ and $E_r=E(G[V_{r}]) \cup E_{\ell \to r}$~.
\item Let $(R_{r},P_{r})=\Direachability(G_r,S,V_r)$ where $G_r=(V_{r} \cup S, E_r)$.
\item For every $(s_i,u'_j)\in E_{\ell \to r}$, update $P_r(i,u'_j)=W_{\ell \to r}(i,j)$.
\item Return $R=R_\ell \circ R_r$ and $P=P_\ell \circ P_r$.
\end{enumerate}
\end{wrapper}

\BL \label{lem:reachability-correctness}
The output matrix $R$ satisfies that $R(i,j)=1$ iff $v_j$ is reachable from $s_i$ under the non-reflexive convention above, for every $i \in \{1,\ldots, |S|\}$ and $j \in \{1,\ldots, n\}$. Moreover, $P$ is the predecessor matrix for $R$.
\EL
\begin{proof}
We show by induction on $i \in \{0,\ldots, \log n\}$ that the claim holds for $|V|\leq 2^i$ vertices. 
For the base case of $i=0$, the only vertices reported as reachable are those with an incoming edge from a source in $S$, and the correctness is immediate.

Assume now that the claim holds up to $i-1$, and consider a set $V'$ with $2^i$ vertices. Let $V_\ell,V_r$ be the partition of $V'$ based on the topological ordering of $V'$. By the induction assumption for $(i-1)$, 
we have that $(R_{\ell}, P_{\ell})=\Direachability(G[V_{\ell}],S,V_{\ell})$ are the correct reachability and predecessor matrices from $S$ to $V_\ell$ in $G$.

The algorithm defines a graph $G_r=(V_{r} \cup S, E(G[V_{r}]) \cup E')$ where $E'$ is a subset of $S \times V_r$ edges satisfying the following: for every $s_i \in S$ and $v'_j \in V_r$, we have $(s_i,v'_j) \in E'$ iff there is an $s_i$ to $v'_j$ path in $G$ whose internal vertices are in $V_\ell$. Hence $G_r$ is $G[V_r]$ augmented with shortcut edges, and it is not merely the induced graph $G[V_r]$. We therefore have
\begin{equation}\label{eq:right-reachability}
s_i \leadsto_{G_r} v'_j \quad\Longleftrightarrow\quad s_i \leadsto_G v'_j
\end{equation}
for every $s_i \in S$ and $v'_j \in V_r$. Indeed, every shortcut edge in $E'$ expands to the corresponding path in $G$ through $V_\ell$, and every $s_i$ to $v'_j$ path in $G$ can be represented in $G_r$ by replacing the prefix that goes through $V_\ell$ with the corresponding shortcut edge. By the induction assumption on $(i-1)$, $R_r$ and $P_{r}$ are the correct reachability and predecessor matrices from $S$ to $V_r$ in $G_r$. By \eqref{eq:right-reachability}, $R_r$ is also the correct reachability matrix for the vertices in $V_r$ in the original graph $G$.

It remains to fix the predecessor matrix, since $G_r$ is not a subgraph of $G$: it contains shortcut edges. Note that $E(G_r) \setminus E(G) = E_{\ell \to r}$. Hence, we only need to fix the predecessors of one-edge $s_i$-$u'_j$ paths $(s_i,u'_j)\in E_{\ell \to r}$. Since $W_{\ell \to r}$ is the witness matrix of $Z_{\ell \to r}=R_{\ell} \times E_G[V_{\ell} \to V_r]$, it is also the predecessor matrix of the positive entries of $Z_{\ell \to r}$. The correction is therefore valid because $P_r(i,j)$ is set to the index of the predecessor vertex of $u'_j$ on some $s_i$ to $u'_j$ path. Consequently, the output matrices $R_\ell \circ R_r$ and $P=P_{\ell}\circ P_r$ are the correct $S \times V$ reachability and predecessor matrices in $G$.
\end{proof}

Let $MM(a,b)$ be the time complexity for computing the matrix multiplication of an $a \times b$ matrix by a $b \times b$ matrix. 

\BL \label{lem:runtime}
Algorithm $\Direachability(G,S,V)$ can be implemented in $O(n^{\omega(\sigma)}\log n+n\cdot |S|)$ time where $n=|V(G)|$ and $|S|=n^{\sigma}$. 
\EL
\begin{proof}[Proof of Lemma \ref{lem:runtime}]
The runtime can be bounded by the recurrence
\[
W(n)=2W(n/2)+MM(|S|,n/2) \quad \mbox{for } n\geq 2,
\qquad
W(1)=MM(|S|,1).
\]
By unrolling the recurrence, we get: 

\begin{equation}\label{eq:reach-recurs}
W(n)=\sum_{j=1}^{\log n} 2^{j} \cdot MM\left(|S|, \frac{n}{2^j}\right) + n \cdot |S|.
\end{equation}
We next bound $MM(|S|, n/2^j)$. To multiply an $|S|\times n/2^j$ matrix by an $n/2^j \times n/2^j$ matrix, split the source set $S$ into\footnote{This holds as $n^{\sigma}\leq 2^j \cdot (n/2^j)^{\sigma}$.} at most $2^j$ subsets, each of size at most $(n/2^j)^{\sigma}$. This yields 

\[
MM\left(|S|,\frac{n}{2^j} \right)\leq 2^{j} \cdot \left(\frac{n}{2^j} \right)^{\omega(\sigma)}.
\]
Plugging this bound into Eq.~\eqref{eq:reach-recurs} gives
\[
W(n)
\leq \sum_{j=1}^{\log n} 2^{j}\cdot 2^{j}\left(\frac{n}{2^j}\right)^{\omega(\sigma)}
    + n\cdot |S|
= n^{\omega(\sigma)}
  \sum_{j=1}^{\log n} 2^{j(2-\omega(\sigma))}
  + n\cdot |S|.
\]
Since $\omega(\sigma)\geq 2$, the sum is at most $\log n$. Therefore,
\[
W(n)=O(n^{\omega(\sigma)} \log n+ n\cdot |S|).
\]

\end{proof}

\paragraph{Handling General Graphs.} For a general digraph, we reduce to a DAG instance $G'$ by contracting each strongly connected component (SCC), and then apply 
Algorithm $\Direachability$ to the contracted graph $G'$ with the source set corresponding to the SCCs that contain the input sources $S$. Given the predecessor matrix $P'$ for the contracted graph $G'$, as in \cite{AlonGMN92}, we open the contracted vertices and transform $P'$ into an $n^{\sigma}\times n$ predecessor matrix for the original graph $G$. For simplicity, let $\{1,\ldots, n\}$ be the IDs of the vertices in $G'$, and let $v_{i,1},\ldots, v_{i,\ell_i}$ be the IDs of the vertices in the $i$-th SCC. We compute a predecessor matrix $\hat{P}$ for the subgraphs induced by the SCCs. This is done by computing incoming and outgoing BFS trees rooted at $v_{i,1}$ for each SCC $i$. For every $v_{i,a}$ and $v_{i,b}$, $\hat{P}(v_{i,a},v_{i,b})$ indicates a predecessor neighbor of $v_{i,b}$ on some $v_{i,a}$ to $v_{i,b}$ path. We now use $P'$ and $\hat{P}$ to define the output $n^{\sigma}\times n$ matrix $P$, as follows. 
For every source $s \in S$, let $i_s$ be the identifier of the SCC containing $s$. Then, for every $P'(i_s,j)=k$, let $(v_{k,a},v_{j,b})$ be a $G$-edge that corresponds to the $(k,j)$ edge in $G'$. Set $P(s,v_{j,b})=v_{k,a}$ and for every $y \in \{1,\ldots, \ell_j\}\setminus \{b\}$, let $P(s,v_{j,y})=\hat{P}(v_{j,b},v_{j,y})$.
This entire computation takes $O(n^{2})$ time. 
The proof of \Cref{thm:multi-source} follows.

\bibliographystyle{alpha}
\bibliography{thesis}

\appendix
%
%
%
%
%
%
%

\end{document}